\documentclass[letterpaper, 10 pt, conference]{ieeeconf}
\IEEEoverridecommandlockouts
\usepackage{cite}
\usepackage{amsmath,amssymb,amsfonts}
\usepackage{algorithm} 
\usepackage{algpseudocodex}
\usepackage{graphicx}
\usepackage{textcomp}
\usepackage{xcolor}
\usepackage{url}
\usepackage{dblfloatfix}

\def\BibTeX{{\rm B\kern-.05em{\sc i\kern-.025em b}\kern-.08em
    T\kern-.1667em\lower.7ex\hbox{E}\kern-.125emX}}
\usepackage{mathtools}

\usepackage{pgfplots}
\pgfplotsset{compat=newest}

\usepackage{acronym}
\acrodef{OLS}{ordinary least squares}
\acrodef{MLE}{maximum likelihood estimator}
\acrodef{SNR}{signal-to-noise ratio}
\acrodef{LTI}{linear time-invariant}
\acrodef{ARX}{autoregressive systems with exogenous inputs}

\newtheorem{theorem}{Theorem}

\newtheorem{remark}{Remark}
\newtheorem{assumption}{Assumption}

\DeclareMathOperator*{\argmax}{arg\,max}
\DeclareMathOperator*{\argmin}{arg\,min}

\DeclareMathOperator*{\Tr}{Tr}

\newcommand{\R}{\mathbb{R}}

\newcommand{\Prob}{\mathbb{P}}
\newcommand{\Expect}{\mathbb{E}}
\newcommand{\N}{\mathcal{N}}

\newcommand{\bigO}{\mathcal{O}}
\newcommand{\simiid}{\stackrel{\text{i.i.d.}}{\sim}}

\newcommand*{\her}{\mathsf{H}}
\newcommand{\Acal}{\mathcal{A}}
\newcommand{\Bcal}{\mathcal{B}}
\newcommand{\rank}[1]{\mathrm{rank({#1})}}

\begin{document}
\title{Hidden Convexity in Active Learning: \\ A Convexified Online Input Design for ARX Systems
\thanks{NC and AI acknowledge the support by Deutsche Forschungsgemeinschaft (DFG, German Research Foundation) under Germany's Excellence Strategy - EXC 2075 – 390740016 and by the Stuttgart Center for Simulation Science (SimTech).
BS acknowledges the support of the International Max Planck
Research School for Intelligent Systems (IMPRS-IS).}
\thanks{The authors are with the University of Stuttgart, Institute for Systems Theory and Automatic Control, 70550 Stuttgart,
Germany (e-mail: \{\mbox{nicolas.chatzikiriakos}, \mbox{bowen.song}, \mbox{andrea.iannelli}\}@ist.uni-stuttgart.de.)}
}

\author{Nicolas Chatzikiriakos, Bowen Song, Philipp Rank, and Andrea Iannelli}

\maketitle
    
\begin{abstract}
The goal of this work is to accelerate the identification of an unknown ARX system from trajectory data through online input design. 
Specifically, we present an active learning algorithm that sequentially selects the input to excite the system according to an experiment design criterion using the past measured data. The adopted criterion yields a non-convex optimization problem, but we provide an exact convex reformulation allowing to find the global optimizer in a computationally tractable way. 
Moreover, we give sample complexity bounds on the estimation error due to the stochastic noise. Numerical studies showcase the effectiveness of our algorithm and the benefits of the convex reformulation.
\end{abstract}

\begin{keywords}
Input Design, Convex Optimization, Finite Sample Identification
\end{keywords}

\section{Introduction}
Input design methods from system identification leverage the flexibility in choosing the excitation input to accelerate the learning process of identifying an unknown system from data. 
However, existing works suffer from relying on computationally expensive or intractable operations~\cite{wagenmaker2020active} or only provide asymptotic guarantees~\cite{bombois2011optimal,gerencser2009identification}.      
To remedy this, in this work, we present a sample efficient algorithm that solves a computationally tractable problem to sequentially perform optimal experiments for fast learning of unknown \ac{ARX}. 
Importantly, the algorithm relies on solving a convex optimization problem to design optimal excitations and is equipped with finite sample guarantees.   
Our work is closely connected to classical experiment design \cite{gerencser2009identification, manchester2010input} and recent works in finite sample system identification \cite{wagenmaker2020active, Ziemann2023}. Next, we review the relevant literature to paint a clearer picture of existing gaps.
\paragraph*{Classical input design}
The design of optimal experiments has a long history in the literature of system identification; see \cite{1977dynamic, bombois2011optimal} for an overview on the topic.
One key problem in experiment design is that the design criteria, which are often a function of the Fisher-information matrix, naturally depend on the unknown system. 
To circumvent this problem, some works focus on a robust approach, typically using minimax objectives (see \cite{rojas2007robust} and the references therein). 
Another research thread focused on adaptively designing inputs based on a running estimate \cite{gerencser2005adaptive}.
In particular, \cite{gerencser2009identification} considered the problem of adaptive input design for \ac{ARX} systems in an asymptotic setting, yielding parameter convergence and an asymptotically optimal covariance matrix. 
To tackle cases where the optimization criterion is non-convex, \cite{manchester2010input} leverages 
convex 
relaxations to convexify the design objective.
However, even though in parts computationally attractive, these works only analyze the asymptotic behavior of the identification algorithms as the number of samples goes to infinity, whereas in practice, only finite data is available.
\paragraph*{Finite sample identification and active learning}
In recent years, due to novel results in high-dimensional statistics \cite{wainwright2019high}, there has been an increased interest in the finite sample analysis of system identification; see \cite{Ziemann2023} for an overview of the topic.
Primarily, the literature focused on analyzing system identification under random excitations for different classes of fully observed dynamical systems such as linear \cite{simchowitz2018learning, tsiamis2021linear}, bilinear \cite{sattar2022finite, chatzikiriakos2024b} and certain classes of nonlinear systems \cite{foster20a, Sattar2022}. 
Further, there have been considerable efforts in the analysis of learning a system description from input-output data \cite{oymak2019non, Ziemann2023}. 
Based on these findings, some works departed from purely random excitation and tried to leverage the excitation input to obtain a more sample-efficient identification.
The first work to address this problem was \cite{wagenmaker2020active}, in which the authors provide a sequential 
algorithm resembling a finite-sample version of E-optimal experiment design.
In \cite{wagenmaker2021task}, the same authors take downstream tasks and objectives into account, enabling task-specific learning.  
Subsequent works \cite{mania2022active, lee2024active} analyze algorithms for input design for non-linear systems under a finite sample perspective. 
However, even for the simplest case of fully observed linear systems, the optimal excitation is determined by solving a non-convex optimization problem. 
Thus, the global optimizer cannot be found efficiently, yielding a computationally unattractive algorithm.
To remedy this, \cite{wagenmaker2021task} proposes a convex relaxation of the problem. 
While this relaxation makes the algorithm computationally tractable, it is not guaranteed to recover the optimal solution of the original, non-convex problem.  
Finally, input design algorithms with finite sample guarantees are not available for \ac{ARX} systems. 
\subsection{Contribution}
We provide a new input design algorithm that sequentially generates periodic inputs to enable fast identification of unknown \ac{ARX} systems affected by Gaussian process noise. 
Our analysis also gives an upper bound on the finite sample estimation error of the unknown parameters. 
To do so, we leverage that the covariates' evolution can be represented by a state-space system, enabling us to use existing tools for the finite sample analysis. 
This results in an upper bound on the identification error of the parameters of the unknown \ac{ARX} system. 
In addition, inspired by considerations in \cite{manchester2010input, wagenmaker2020active}, we propose a convex reformulation to leverage a hidden convexity in the input design objective. 
We do so by transforming the non-convex optimization problem into a convex one from which the \emph{global} optimizer of the original non-convex objective can be reconstructed. 
Thus, the proposed algorithm only requires solving a convex optimization problem online.
Importantly, the convex reformulation deployed is not restricted to \ac{ARX} systems, but can also be used for fully observed linear systems and hence is of independent interest.   
With this, we provide a computationally tractable input design scheme equipped with finite sample guarantees, and therefore bridge some of the gaps between the active learning and classical input design literature. 
In concluding simulations, we demonstrate the effectiveness of our algorithm by comparing its estimation error to random excitations and the related approach from \cite{wagenmaker2021task}. 
The experiments show that our algorithm achieves the lowest estimation error, thus enhancing learning, while being computationally attractive and giving a guaranteed estimation accuracy. 
\subsection{Notation}
Given a matrix $M$ we denote its spectral radius, smallest eigenvalue, and spectral norm by $\rho(M)$, $\lambda_\mathrm{min}(M)$, and $\Vert M\Vert$, respectively.
We denote universal constants by $C$.
The set of all non-negative integers (the set of all non-negative integers smaller than $k$) is indicated by $\mathbb{N}$ ($[k]$).
We use $\mathrm{S}^n_+$ to denote the cone of positive semi-definite matrices of dimension $n$.
We use $[a]_i$ ($[A]_{ij}$) when referring to the $i$-th ($ij$-th) element of a vector $a$ (matrix $A$). 
Given matrices $A_i$, $i\in[1, p]$ we use $A_{1:p} = \begin{bmatrix}
    A_1 &\dots & A_p
\end{bmatrix}$.
We write $A^\her$ to denote the conjugate transpose of a complex matrix $A$.
Given a periodic time-domain signal $\{x_t\}_{t=1}^k$ we express its Fourier coefficients with $\check{x}_t$ and the Fourier transform by $\mathcal{F}\{\cdot\}$.
\section{Preliminaries and problem statement}
We consider the problem of identifying an unknown \ac{ARX} system of known orders $p, q >0$ described by
\begin{equation}\label{eq:ARX_System}
    y_t =  \sum_{i=1}^{p} A^*_i y_{t-i} + \sum_{j=1}^{q} B^*_j u_{t-j} + \Sigma_w^{\frac12}w_t,
\end{equation}
where $y_t \in \R^{n_y}$ and $u_t\in \R^{n_u}$ are the output and input of the system at timestep $t$, and $w_t\simiid \N(0, I_{n_x})$ is normalized process noise. We assume that the noise covariance matrix $\Sigma_w$ is known.
As initial conditions we assume $y_{-i} = 0$, $ u_{-i} = 0$, for integers $i \in[ 1, q]$. 
The goal of this work is to identify the unknown system matrices $\{A^*_i\}_{i=1}^p$, $\{B^*_j\}_{j=1}^q$ from data $\{y_t\}_{t=0}^{T}$, $\{u_t\}_{t=0}^{T-1}$ collected from a single trajectory. 
To this end, set $n_x \coloneqq (pn_y + qn_u)$ and define $\theta^*\coloneqq \begin{bmatrix}
    A_{1:p}^* & B_{1:p}^*
\end{bmatrix} \in \R^{n_y \times n_x}$ and $x_t \coloneqq \begin{bmatrix}
    y_{t-1:t-p}^\top  & u_{t-1:t-q} ^\top  
\end{bmatrix}^\top \in \R^{n_x}$ to rewrite~\eqref{eq:ARX_System} as 
\begin{equation}
    y_t = \theta^* x_t + \Sigma_w^{\frac12} w_t.
\end{equation}
To identify the unknown matrix $\theta^*$ we use the \ac{OLS} estimator 
\vspace{-8pt}
\begin{equation}\label{eq:OLS}
    \hat \theta_T = \argmin_{\theta\in \R^{n_y \times n_x}} \frac{1}{T} \sum_{t=1}^T \Vert y_t - \theta x_t \Vert_2^2.
\end{equation}
Assuming that the inverse of the empirical covariance matrix 
\vspace{-8pt}

\begin{equation} \label{eq:empCov}
    \hat{\Sigma}_T \coloneqq\frac1T\sum_{t=1}^Tx_t x_t^\top
\end{equation} exists, the \ac{OLS}~\eqref{eq:OLS} admits a well-known, unique closed-form solution, and its estimation error is given by
\vspace{-8pt}
\begin{equation}\label{c}
    \hat \theta_T - \theta^* = \sum_{t=1}^T w_t x_t^\top \Big(\sum_{t=1}^Tx_t x_t^\top\Big)^{-\frac12} \Big(\sum_{t=1}^Tx_t x_t^\top\Big)^{-\frac12}.
\end{equation}
As derived in \cite{Ziemann2023}, the evolution of the covariates can be represented by the non-minimal process
\begin{equation}\label{eq:sysExtended}
    x_{t+1} = \mathcal{A}^*x_t + \mathcal{B}_u u_t + \mathcal{B}_w w_t,
\end{equation}
where     $\Acal^* \coloneqq\begin{bmatrix}
    \Acal^*_{11} & \Acal^*_{12} \\ 0 & \Acal_{22}
\end{bmatrix}$ and 
\begin{subequations}
\begin{align}
\mathcal{A}_{11}^* &\coloneqq \begin{bmatrix}
            A_1^* & A_2^* &\dots & A_p^* \\
            I & 0 & \dots &0 \\
            0 & I & \dots & 0 \\
            0 &\dots &I &0
        \end{bmatrix}, \\
\Acal_{12}^* &\coloneqq \begin{bmatrix}
    1 &0 &\dots &0
\end{bmatrix} ^\top  \otimes B_{1:q}^*, \\
\Acal_{22} &\coloneqq \begin{bmatrix}
    0 &\dots &\dots &0 \\
    I_{n_u} & 0 &\dots &0 \\
    0 & \ddots & 0 &0 \\
    0 & \dots & I_{n_u} &0 
\end{bmatrix}, \\
\Bcal_u &\coloneqq \begin{bmatrix}
    0_{n_u \times p n_y  }& I_{n_u} & 0_{n_u \times (q-1)n_u}
\end{bmatrix}^\top, \\ 
\Bcal_w &\coloneqq \begin{bmatrix}
    \Sigma_w^{\frac{1}{2}} &n_y \times 0_{((p-1)n_y + q n_u)}
\end{bmatrix}^\top.
\end{align}
\end{subequations}
Subsequently, when writing $\hat \Acal_T$ we refer to the matrix $\Acal$ where all unknown quantities have been replaced by the estimate $\hat \theta_T$. 
Note that the matrices $\Bcal_u$ and $\Bcal_w$ are knowns.
In line with the relevant literature \cite{gerencser2009identification, wagenmaker2020active} we assume that the \ac{ARX} system~\eqref{eq:ARX_System} is asymptotically stable.
\begin{assumption}\label{ass:stable}
    System \eqref{eq:ARX_System} is asymptotically stable, i.e., $\rho(\mathcal{A}_{11}) < 1$.
\end{assumption}
While Assumption~\ref{ass:stable} might seem restrictive, related works providing a finite sample analysis of system identification with random inputs are restricted to marginally stable systems \cite{simchowitz2018learning, Ziemann2023}, and learning unstable systems using the \ac{OLS} has been shown to be statistically inconsistent \cite{sarkar2019near}. 
The slightly stronger requirement of asymptotic stability is necessary for the proposed active learning algorithm (see \cite[Remark B.5]{wagenmaker2020active}) and cannot easily be circumvented. 

We now define the controllability Gramians with respect to the control input and the process noise, respectively
\begin{subequations}
    \begin{align}
    \Gamma_t(\Acal, \Bcal_u) \coloneqq \sum_{s=0}^{t-1} \Acal ^s \Bcal_u \Bcal_u^\top {\Acal^s}^\top, \\
    \Gamma_t(\Acal, \Bcal_w) \coloneqq \sum_{s=0}^{t-1} \Acal ^s \Bcal_w \Bcal_w^\top {\Acal^s}^\top.
    \end{align}	
\end{subequations}
Given a periodic input $u_t$ of period $k$ and bounded energy $\sum_{t=0}^{k-1} u_t^\top u_t\le \gamma^2 k$ we denote the normalized steady-state covariance when $w_t \equiv 0$, as
\begin{align}
     &\Gamma^{\check{u}}_k(\mathcal{A}, \Bcal_u) \coloneqq \lim_{T\to \infty} \frac{1}{T \gamma^2}\sum_{t=1}^T x_t x_t^\top \label{eq:Gramianu}\\ 
     &\quad= \frac{1}{\gamma^2 k^2} \sum_{\ell=0}^{k-1} (e^{j \frac{2\pi \ell}{k}} I - \Acal)^{-1} \Bcal_u \check{u}_\ell \check{u}_\ell^\her \Bcal_u^\her (e^{j \frac{2\pi \ell}{k}} I - \Acal)^{-\her}, \notag
\end{align}
where the equality holds due to Parseval's Theorem. 
Given matrices $\{\check{U}_\ell \}_{\ell\in [k]} \in \mathrm{S}_+^{n_u}$, with abuse of notation, we define 
\begin{align}
     &\Gamma^{\check{U}}_k(\mathcal{A}, \Bcal_u)  \label{eq:GramianU}\\
     &\quad  \coloneqq \frac{1}{\gamma^2 k^2} \sum_{\ell=0}^{k-1} (e^{j \frac{2\pi \ell}{k}} I - \Acal)^{-1} \Bcal_u \check{U}_\ell \Bcal_u^\her (e^{j \frac{2\pi \ell}{k}} I - \Acal)^{-\her} . \notag
\end{align}
It is important to observe at this point that \eqref{eq:Gramianu} can be equivalently represented by \eqref{eq:GramianU} by taking $\check{U}_\ell = \check{u}_\ell \check{u}_\ell^\her$, but the opposite is only true if $\rank{\check{U}_\ell} = 1$.
\section{Active learning for \ac{ARX} systems}
In \cite{Ziemann2023}, it was shown that isotropic Gaussian excitations are sufficient to identify the unknown \ac{ARX} system~\eqref{eq:ARX_System} from trajectory data.
However, in most cases isotropic Gaussian excitations are not optimal in terms of the resulting sample complexity. 
To improve the sample complexity of identifying~\eqref{eq:ARX_System} we propose to sequentially generate informative data using Algorithm~\ref{alg:InputDesignAlgo}, which is inspired by the algorithm proposed in \cite{wagenmaker2020active} for fully observed linear systems. 
\begin{algorithm}[H]
\caption{Active Learning Algorithm}
\label{alg:InputDesignAlgo}
\begin{algorithmic}[1]
    \Require $\gamma^2$, $T_0$, $k_0$
    \State Collect dataset $\mathcal{D}_0=\left\{\{y_t\}_{t=0}^{T_0}, \{u_t\}_{t=0}^{T_0}\right\}$ using $u_t \simiid \N(0, \tfrac{\gamma^2}{n_u}I_{n_u})$
    \State Estimate $\hat \theta_0$ from $\mathcal{D}_0$ by solving \eqref{eq:OLS} 
    \State Set $T = T_0$, $\mathcal{D} = \mathcal{D}_0$
    \For{i = 1,2,\dots}
        \State Set $T_{i}= 3T_{i-1}$, $k_i = 2 k_{i-1}$ 
        \State $u^* = \Call{OptInput}{T, T_i, \hat{\theta}_i, k_i, \tfrac{\gamma^2}{2}, \mathcal{D}}$
        \State Set $u^\eta_t = u_t^* + \eta_t$, where $\eta_t \simiid \N(0, \frac{\gamma^2}{2n_u}I_{n_u})$
        \State Collect $\mathcal{D}_{i}=\left\{\{y_t\}_{t=T}^{T+T_i}, \{u^\eta_t\}_{t=T}^{T+T_i}\right\}$ using $u^\eta_t$
        \State Estimate $\hat \theta_i$ from $\mathcal{D} = \bigcup_{j=0}^{i} \mathcal{D}_j$ by solving \eqref{eq:OLS} 
        \State Set $T = T+ T_i$
    \EndFor
\end{algorithmic}
\end{algorithm}
In essence, Algorithm~\ref{alg:InputDesignAlgo} sequentially generates periodic signals that excite the true system \eqref{eq:ARX_System} according to some optimality criterion.
Hereby, the length of each episode $T_i$ is increased exponentially to put more weight on inputs generated with more accurate estimates. Similar strategies have also been used successfully in other sequential learning problems, such as best-arm identification~\cite{karnin2013}.
In our work, we consider periodic input signals since, as shown in~\cite{wagenmaker2020active}, this class is general enough to achieve an optimal rate.
The number of excitation frequencies $k_i$ is increased sequentially, although at a slower rate than $T_i$. The obtained periodic signal is repeated until the end of the episode.
For technical reasons, the algorithm requires to check whether all frequencies are safe to plan with, which is why in the original formulation an uncertainty estimate $\epsilon_i$ is obtained in each round. Empirically, this step is not necessary (see \cite[Remark 2.5]{wagenmaker2020active} for a detailed discussion) and hence is omitted here for clarity of exposition. However, it is important to note that this step is necessary for the proof of Theorem~\ref{th:FiniteSampleActive}.
As information criterion driving input selection, we take the viewpoint of E-optimality \cite{rojas2007robust, bombois2011optimal}. That is, the input sequence is selected to maximize the smallest eigenvalue of the empirical covariance matrix~\eqref{eq:empCov}, while satisfying an energy constraint on the input.
As shown in \cite{wagenmaker2020active}, this translates to optimizing the Fourier coefficients of the optimal input sequence via the non-convex optimization problem
\vspace{-15pt}

\begin{subequations}
    \label{eq:optProbNonConvex}
    \begin{alignat}{2}
    \check{u}^* \in &\argmax_{\{\check{u}_\ell\}_{\ell\in [k_i]}}  \; &&\lambda_{\mathrm{min}} \Big( \frac{\gamma^2}{2} T_i \Gamma^{\check{u}}_{k_i}(\hat{\mathcal{A}}_i, \mathcal{B}_u) + \sum_{t=1}^T x_t x_t^\top \Big) \label{eq:optProbNonConvex1}\\
       & \quad \mathrm{s.\,t.}  && \sum_{\ell \in [k_i]} \Tr(\check{u}_\ell \check{u}_\ell^\her) \le \frac{k_i^2\gamma^2}{2}.\label{eq:optProbNonConvex2}
    \end{alignat}
\end{subequations}
The matrix $\Gamma_{k_i}^{\check{u}}(\hat \Acal_i, \Bcal_u)$ in the optimality criterion~\eqref{eq:optProbNonConvex1}  describes the predicted future effect of the designed excitation input on the covariance (\ref{eq:Gramianu}) and $\sum_{t=1}^T x_t x_t^\top$ is the empirical covariance based on the past measurements collected in $\mathcal{D}$. 
The constraint~\eqref{eq:optProbNonConvex2} is necessary to ensure the energy bound for the input. 
We observe that \eqref{eq:optProbNonConvex} is non-convex in the Fourier coefficients of the unknown input signal. After solving~\eqref{eq:optProbNonConvex} the designed input is computed using the inverse Fourier transform
\begin{equation}\label{eq:DefOPtInput}
        \textsc{OptInput}(T, T_i, \hat \theta_i, k_i, \tfrac{\gamma^2}{2}, \mathcal{D}) = \mathcal{F}^{-1}\left\{\check{u}^*\right\}.
\end{equation}
Since $\theta_*$ is unknown to the learning algorithm, the estimate $\hat \theta_i$ is used to predict the future excitation in the optimization. 
\par 
In the subsequent Section~\ref{sec:FiniteSample}, we show that Algorithm~\ref{alg:InputDesignAlgo} enjoys finite sample guarantees for the identification of \ac{ARX} systems when the input is chosen as the solution of \eqref{eq:optProbNonConvex}-\eqref{eq:DefOPtInput}, which as noted earlier is a non-convex optimization problem. 
This is followed by Section~\ref{sec:convexRelaxation}, where we provide an exact convex reformulation of~\eqref{eq:optProbNonConvex}, so that the optimal excitation can be determined efficiently and reliably.
%
%
\subsection{Finite sample identification of ARX Systems}\label{sec:FiniteSample}
To simplify our presentation, we make the following additional assumption.
\begin{assumption}\label{ass:hyperparameters}
    The hyperparameters $T_0$ and $k_0$ are chosen such that $T_0 \ge f_{\theta_*}(k_0)$ for some known function $f_{\theta_*}(k_0)$.
\end{assumption}
We do not provide the exact expression for $f_{\theta_*}$ here to streamline the discussion and refer to \cite[Appendix B]{wagenmaker2020active} for the definition. Note that, even though $f_{\theta_*}$ depends on $\theta_*$ Assumption~\ref{ass:hyperparameters} is not restrictive, since the sequence $\{T_i\}_{i\ge0}$ increases faster than $\{k_i\}_{i\ge0}$. 
Thus, even if Assumption~\ref{ass:hyperparameters} is not satisfied for $i=0$ it will eventually be satisfied; see \cite[Appendix B]{wagenmaker2020active} for a more detailed discussion.
The sample complexity analysis of Algorithm~\ref{alg:InputDesignAlgo} relies on the system to reach steady-state for the given input sequence. To characterize the transient behavior of process~\eqref{eq:sysExtended} consider 
\vspace{-3pt}
\begin{equation}
    \beta(\mathcal{A^*}) \coloneqq \sup_{k\in \mathbb{N}} \Vert {\mathcal{A}^*}^k\Vert \left(\tfrac{1}{2} + \tfrac{\rho(\mathcal{A}^*)}{2}\right)^{-k}.
\end{equation}
As discussed in \cite{wagenmaker2020active}, $\beta(\mathcal{A}^*)$ is finite when Assumption~\ref{ass:stable} holds.
We are now ready to provide our first main result.
\begin{theorem}\label{th:FiniteSampleActive}
    \begin{figure*}[btp]
    \small{
        \begin{equation}\label{eq:FiniteSampleBoundUpper}
            \Prob\left[\Vert \hat{\theta}_T - \theta_* \Vert \le C \frac{\sqrt{\log\frac{1}{\delta} + n_x + \log \det\left(\bar \Gamma_T \left(\Gamma_{k(T)}(\mathcal{A}^*,\mathcal{B}_w^*) + \frac{\gamma^2}{p} \Gamma_{k(T)}(\Acal^*, \Bcal_u^*)\right)^{-1} + I_{n_x} \right)}}{\sqrt{T\lambda_\mathrm{min}\left(\Gamma_{k(T)}(\mathcal{A}^*,\mathcal{B}_w^*) + \gamma^2 \Gamma_{k(T)}^{\check{u}^*}(\mathcal{A}^*, \mathcal{B}^*_u)\right)} }\right]\ge 1 -\delta
        \end{equation}}
        \medskip
        \vspace*{-0.6\baselineskip}
        \hrule
        \medskip
        \vspace*{-0.6\baselineskip}
    \end{figure*}
    Let Assumptions~\ref{ass:stable} and~\ref{ass:hyperparameters} hold and fix some $\delta\in(0,1)$. Then Algorithm~\ref{alg:InputDesignAlgo} will produce inputs that satisfy $\Expect\left[\sum_{t=1}^T(u^{\eta}_t)^\top u^{\eta}_t\right] \le \gamma^2 T$. If further $T \ge \bar{T}(\theta_*, \frac{1}{\delta}, k(T))$, where $\bar{T}(\theta_*, \tfrac1\delta, k(T))$ is defined in \cite[Appendix B]{wagenmaker2020active}, the estimate $\hat \theta_T$ produced by Algorithm~\ref{alg:InputDesignAlgo} will satisfy~\eqref{eq:FiniteSampleBoundUpper}, where
    \begin{equation}
        \Bar \Gamma_T = \bigO\left(\frac{\beta(\Acal^*)^2\gamma^2 T}{ (1-\rho(\Acal^*))^2}\right) I_{n_x} 
    \end{equation}
    and $\check{u}^* = \mathcal{F}\{\textsc{OptInput}(T, \theta^*, k(T), \gamma^2 ,\emptyset)\}$.
\end{theorem}
\begin{proof}
    First, note that by \eqref{eq:OLS} the estimation error can be bounded as  
    \begin{equation}
    \begin{aligned}
        \Vert \hat \theta_T - \theta^*\Vert \le &\lambda_\mathrm{min}\Big(\sum_{t=1}^Tx_t x_t^\top\Big)^{-\frac12} \\  &\times\left\Vert\sum_{t=1}^T w_t x_t^\top \Big(\sum_{t=1}^Tx_t x_t^\top\Big)^{-\frac12} \right\Vert.
    \end{aligned}
    \end{equation}
    Recall that the evolution of the covariates follows \eqref{eq:sysExtended}. 
    Since $\Acal^*$ is block-diagonal we have 
    \begin{equation}
        \rho(\Acal^*) = \max\Big(\rho(\Acal_{11}^*), \rho(\Acal_{22})\Big) =  \rho(\Acal_{11}^*) < 1,
    \end{equation} where the second equality follows since $\rho(\Acal_{22}) = 0$ by the lower triangular structure and the inequality follows from Assumption~\ref{ass:stable}.
    With this, we can use the results from \cite{wagenmaker2020active} to obtain Theorem~\ref{th:FiniteSampleActive}.
    Specifically, the analysis and the active learning algorithm use \eqref{eq:sysExtended} to quantify the influence of the control input and process noise on the covariates. 
\end{proof}
\begin{remark}
     For clarity in Theorem~\ref{th:FiniteSampleActive} we only provide a finite sample error upper bound for Algorithm~\ref{alg:InputDesignAlgo}. However, the optimality analysis in \cite{wagenmaker2020active} can also be adapted for \ac{ARX} systems by using the same rationale as in the proof of Theorem~\ref{th:FiniteSampleActive}.
\end{remark}
%
%
\subsection{Convexified input design}\label{sec:convexRelaxation}
Up to this point, Algorithm~\ref{alg:InputDesignAlgo} relied on solving the non-convex optimization problem~\eqref{eq:optProbNonConvex}.
However, this is often intractable and might result only in \emph{local} optima. 
This is an issue because the theoretical guarantees in Theorem 1 only hold when the solution of \eqref{eq:optProbNonConvex}-\eqref{eq:DefOPtInput} is used, and the empirical performance of Algorithm 1 degrades when suboptimal solutions are used.
To circumvent this, we propose an exact convex reformation of~\eqref{eq:optProbNonConvex} and show that the optimal input can be obtained through a convex optimization problem, precisely, a Semidefinite Program.
\begin{theorem}\label{th:convexRelax}
    Let $\check{U}_\ell^*$, $\ell \in [k_i]$  be the solution of 
    \begin{subequations}
    \label{eq:optProbConvex}
    \begin{alignat}{2}
        &\max_{\{\check{U}_\ell\}_{\ell \in [k_i]} \in \mathrm{S}_+^{d_u}} &&\lambda_{\mathrm{min}} \left( \tfrac{\gamma^2}{2}  	T_i\Gamma^{\check{U}}_{k_i}(\hat{\mathcal{A}}_i,{\mathcal{B}}_u) + \sum_{t=1}^T x_t x_t^\top \right) \label{eq:optProbConvex1}\\
        &\qquad \mathrm{s.\,t.} \quad && \sum_{\ell \in [k_i]} \Tr(\check{U}_\ell) \le \frac{{k_i}^2 \gamma^2}{2} \label{eq:optProbConvex2}
    \end{alignat}
    \end{subequations}
    and let $\check{u}_\ell^*$ be the eigenvector corresponding to the largest eigenvalue of $\check{U}_\ell^*$ for each $\ell \in [k_i]$. 
    Then $\check{u}_\ell^*$, $\ell \in [k_i]$ is the solution of~\eqref{eq:optProbNonConvex}.
\end{theorem}
\begin{proof}
    In this proof, we first show that \eqref{eq:optProbNonConvex} can be equivalently formulated as \eqref{eq:optProbConvex} with an additional non-convex rank constraint. 
    We proceed by analyzing the objective arising in both optimization problems.
    Then, we analytically derive the optimizers of~\eqref{eq:optProbNonConvex} and~\eqref{eq:optProbConvex} and use them to show the result.
    \\
    \textit{Part 1: Equivalent reformulation of~\eqref{eq:optProbNonConvex}:}
    Given a vector $\check{u}_\ell$, we define $\check{U}_\ell \coloneqq \check{u}_\ell \check{u}_\ell^\her$ which is rank-$1$. This allows us to rewrite the optimization problem \eqref{eq:optProbNonConvex} as
    \vspace{-7pt}
    \begin{subequations}
    \label{eq:optProbNonConvexre}
    \begin{alignat}{2}
        &\max_{\{\check{U}_\ell\}_{\ell \in [k_i]} \in \mathrm{S}_+^{d_u}} &&\lambda_{\mathrm{min}} \Big(\frac{\gamma^2}{2} T_i \Gamma^{\check{U}}_{k_i}(\hat{\mathcal{A}}_i,{\mathcal{B}}_u) + \sum_{t=1}^T x_t x_t^\top \Big) \label{eq:optProbNonConvexre1}\\
        &\qquad \mathrm{s.\,t.} \quad && \sum_{\ell\in [k_i]} \Tr(\check{U}_\ell) \le \frac{{k_i}^2\gamma^2}{2}\label{eq:optProbNonConvexre2}\\
        &\qquad \quad &&\mathrm{rank}(\check{U}_\ell)=1,\quad \forall \ell\in [k_i].\label{eq:optProbNonConvexre3}
    \end{alignat}
    \end{subequations}
   Note that the only difference between \eqref{eq:optProbNonConvexre} and \eqref{eq:optProbConvex} is the non-convex rank constraint \eqref{eq:optProbNonConvexre3}.
   In particular, the costs \eqref{eq:optProbConvex1} and \eqref{eq:optProbNonConvexre1} are identical and can be analyzed jointly.
   \\
   \textit{Part 2: Analyzing the unconstrained problem:}
    We define
    \begin{equation}
        M(\check{U}):=\frac{\gamma^2}{2} T_i \Gamma^{\check{U}}_{k_i}(\hat{\mathcal{A}}_i, {\mathcal{B}}_u) + \sum_{t=1}^T x_t x_t^\top\succeq 0
    \end{equation}
    and recall that $ \lambda_\mathrm{min}(M(\check U))= \min_{\Vert v \Vert =1} v^\her M(\check U) v $.
    Defining $F_\ell \coloneqq (e^{j \frac{2\pi \ell}{k_i}} I - \hat{\Acal}_i)^{-1} {\Bcal}_u$ for $\ell \in[k_i]$, $c_i:=\frac{T_i}{2 k_i^2}$ and using elementary reformulations we obtain  
    \begin{equation}
    \begin{aligned}
         \lambda_\mathrm{min}&(M(\check U)) \\ &= \min_{\Vert v \Vert =1} c_i v^\her \Big(\sum_{\ell\in [k_i]}F_\ell \check U_\ell F_\ell^\her \Big)v + v^\her \sum_{t=1}^T x_t x_t^\top v.
    \end{aligned}
    \end{equation}
    With this we can rewrite~\eqref{eq:optProbNonConvexre1} as
    \begin{equation}\label{Convexproof2}
    \begin{aligned}
        \max_{\{\check{U}_\ell\}_{\ell \in [k_i]} \in \mathrm{S}_+^{d_u}} \min_{\Vert v \Vert =1}\Bigg[c v^\her \Big(\sum_{\ell\in [k_i]}&F_\ell \check U_\ell F_\ell^\her \Big)v \\ &+ v^\her \sum_{t=1}^T x_t x_t^\top v\Bigg].
    \end{aligned}
    \end{equation}
    Let $v_*$ be the vector that achieves the minimum in~\eqref{Convexproof2}. 
    Note that for a fixed $v_*$, the maximum in \eqref{Convexproof2} is independent of the second term of the sum in \eqref{Convexproof2}.
    Defining $H_\ell:=F_\ell^\her v_* v_*^\her F_\ell$,  the optimal solution to \eqref{Convexproof2} is identical to the solution to 
    \begin{equation}\label{eq:transform1}
    \begin{aligned}
       \max_{\{\check{U}_\ell\}_{\ell \in [k_i]} \in \mathrm{S}_+^{d_u}}  &\sum_{\ell \in [k_i]} v_*^\her F_\ell \check U_{\ell} F_\ell^\her v_* \\
       &=\max_{\{\check{U}_\ell^*\}_{\ell \in [k_i]} \in \mathrm{S}_+^{d_u}}  \sum_{\ell \in [k_i]} \Tr \left(H_\ell\check U_{\ell}\right),
    \end{aligned}
    \end{equation}
    where the equality follows from the cyclic property of the trace. 
    Note that each matrix $H_\ell\succeq 0$, $\ell\in[k_i]$, admits an eigendecomposition $H_\ell=V_{\ell}\Lambda_{\ell} V_{\ell}^\her$
    with $V_{\ell} V_{\ell}^\her=I$. With this, it holds that 
    \begin{equation}\label{eq:rewrittenCost}\begin{aligned}
         &\max_{\{\check{U}_\ell\}_{\ell \in [k_i]} \in \mathrm{S}_+^{d_u}}  \sum_{\ell \in [k_i]} \Tr \left(H_\ell\check U_{\ell}\right) \\ &\qquad \qquad = \max_{\{\check{U}_\ell\}_{\ell \in [k_i]} \in \mathrm{S}_+^{d_u}}  \sum_{\ell \in [k_i]} \Tr \left(V_{\ell}\Lambda_{\ell} V_{\ell}^\her \check{U}_\ell\right) .
    \end{aligned}
    \end{equation}

    \textit{Part 3: Analytical solution of \eqref{eq:optProbNonConvexre}:}
    By the constraint~\eqref{eq:optProbNonConvexre3} there exist vectors $u_\ell$ such that $\check{U}_\ell = u_\ell u_\ell^\her$ for all $\ell \in [k_i]$, so that we can rewrite~\eqref{eq:rewrittenCost} as
    \begin{equation}\label{eq:transform2}
    \begin{aligned}
    \max_{\{\check{U}_\ell\}_{\ell \in [k_i]} \in \mathrm{S}_+^{d_u}}  &\sum_{\ell \in [k_i]} \Tr \left(V_{\ell}\Lambda_{\ell} V_{\ell}^\her \check{U}_\ell\right) \\ &= \max_{\{u_\ell\}_{\ell \in [k_i]}} \sum_{\ell \in [k_i]} u_\ell^\her V_\ell \Lambda_\ell V_\ell^\her u_\ell.  
    \end{aligned}
    \end{equation}
    Applying the change of variables $z_\ell=V_{\ell}^\her u_\ell$ we equivalently rewrite \eqref{eq:transform2} as
    \begin{equation}\label{eq:transform4}
        \max_{\{z_\ell\}_{\ell \in [k_i]}}\sum_{\ell\in [k_i]}z_\ell^\her\Lambda_{\ell} z_\ell=\max_{\{z_\ell\}_{\ell \in [k_i]}} \sum_{\ell\in [k_i]}\sum_{j=1}^{n_u} [z_\ell]_{j}[z_\ell^\her]_{j}[\Lambda_{\ell}]_{jj}.
    \end{equation}
    Since the sum is separable, the optimal strategy under the energy constraint~\eqref{eq:optProbNonConvexre2} is to allocate all the weight to the largest eigenvalue across all $\Lambda_{\ell}$, $\ell \in[k_i]$. 
    Note that if the largest eigenvalue has multiplicity $m > 1$, there exist multiple optimal solutions that distribute the energy across the largest eigenvalues.
   \textit{Part 4: Analytical solution of \eqref{eq:optProbConvex}:} 
   Using the cyclic property of the trace and the change of variables $Z_\ell=V_{\ell}^\her \check{U}_\ell V_{\ell},\forall \ell\in[k_i]$, we can equivalently rewrite~\eqref{eq:rewrittenCost} as 
    \begin{equation}\label{eq:transform3}
        \begin{split}
            \max_{\{Z_\ell\}_{\ell \in [k_i]}\in \mathbb{S}^{n_u}_+} &\sum_{\ell\in [k_i]}\Tr(H_\ell V_{\ell} Z_{\ell} V_{\ell}^\her)\\
            =& \max_{\{Z_\ell\}_{\ell \in [k_i]}\in \mathbb{S}^{n_u}_+} \sum_{\ell\in [k_i]} \sum_{j=1}^{n_u} \left([\Lambda_{\ell}]_{jj} [Z_\ell]_{jj} \right),
        \end{split}
    \end{equation}
   where the second equality holds because $\Lambda_{\ell}$ is diagonal. Under the energy constraint \eqref{eq:optProbConvex2}, the optimal allocation strategy remains to allocate all weight to the largest eigenvalue across all $\Lambda_{\ell}$ with $\ell \in [k_i]$, or if multiple largest eigenvalues exist, to distribute the weight among them.If the latter is true, multiple optimal solutions exist. 
   Thus, the optimal solutions $Z^*_\ell$ for $\ell\in[k_i]$ are diagonal matrices with rank from $1$ to at most $m$, where $m$ is the number of the largest eigenvalues across all $\Lambda_\ell$ for $\ell \in[k_i]$. It is important to emphasize that given $v_*$ a rank-$1$ solution always exists.  \\
   \textit{Part 5: Conclusion:}
   Let $z^*_\ell$ and $Z^*_\ell$ be the optimal solutions to \eqref{eq:transform4} and \eqref{eq:transform3}, respectively. Specially, let $Z^1_\ell$ denote the optimal rank-$1$ solution to \eqref{eq:transform3}. Then, we have
   \begin{equation*}
\sum_{\ell\in [k_i]}{z^*_\ell}^\her\Lambda_{\ell} z^*_\ell=\sum_{\ell\in [k_i]}\Tr\left(\Lambda_{\ell} Z^1_\ell \right)= \sum_{\ell\in [k_i]}\Tr\left(\Lambda_{\ell} Z^*_\ell \right).
   \end{equation*}
   Applying the inverse change of variables $u_\ell^*=V_\ell z_\ell$, $\check{U}^*_\ell=V_\ell Z^*_\ell V_\ell^\her$ and $\check{U}^1_\ell=V_\ell Z^1_\ell V_\ell^\her$, we have
   \begin{equation*}
    \begin{split}
        \sum_{\ell\in [k_i]} {u^*_\ell}^\her H_\ell u^*_\ell=\sum_{\ell\in [k_i]} \Tr \left(H_\ell\check U^1_{\ell}\right)= \sum_{\ell\in [k_i]} \Tr \left(H_\ell\check U^*_{\ell}\right).
    \end{split} 
    \end{equation*}  
    As established in Part 4, any solution $\{\check{U}_\ell^*\}_{\ell \in [k_i]}$ of the convex problem \eqref{eq:optProbConvex} consists of matrices with rank at most $m$. Take an orthogonal eigendecomposition of $\{\check{U}_\ell^*\}_{\ell \in [k_i]}$ and denote any eigenvector corresponding to a nonzero eigenvalue by $\{\bar{u}_\ell^*\}_{\ell\in [k_i]}$. Then the optimal solution $\sqrt{\Tr({\check{U}^*_\ell})}\bar{u}_{\ell}^*$, $\forall \ell \in [k_i]$ achieves the same objective as in \eqref{eq:optProbNonConvexre}.
    \end{proof}
Theorem~\ref{th:convexRelax} shows that the non-convex optimization problem~\eqref{eq:optProbNonConvex}  used in line 6 of Algorithm~\ref{alg:InputDesignAlgo} has a hidden convexity property that we leverage to obtain an exact convex reformulation. 
It is worth observing that Theorem~\ref{th:convexRelax} also applies to fully observed state space systems, hence, the convex reformulation can also be used for the algorithm presented in~\cite{wagenmaker2020active}, yielding computationally attractive input design algorithms for both setups.
\section{Numerical Example}
To showcase the effectiveness of the proposed algorithm, we compare its estimation error with the error resulting from purely isotropic Gaussian excitations and the approach proposed in \cite{wagenmaker2021task}, which also relies on solving a convex problem, but
is not guaranteed to recover the solution of \eqref{eq:optProbNonConvex}.
The key difference of our approach compared to \cite[Algorithm 3]{wagenmaker2021task}, which also proposes a convex input design, is that the latter sequentially applies all eigenvectors of the solutions of the convex problem~\eqref{eq:optProbConvex} instead of just the largest eigenvector. 
As shown in \cite[Proposition 6.5]{wagenmaker2021task}, this reproduces any optimal solution to~\eqref{eq:optProbConvex} as $T\to \infty$, i.e., only as the number of samples goes to infinity. 
We consider the following \ac{ARX} system
\begin{align*}
    A_1^* &= \begin{bmatrix}
        0.7 & 0.1\\
        0 & 0.9  
    \end{bmatrix}, \,
    A_2^* =  \begin{bmatrix}
        -0.5 & 0 \\
        0.1 & -0.2 
    \end{bmatrix},
    \,
    B_1^* = \begin{bmatrix}
        0.1 & 0\\
        0 & 5 
    \end{bmatrix},
\end{align*}
which satisfies Assumption~\ref{ass:stable}. Further, we fix $\Sigma_w = I_{n_x}$ and $\gamma= 10$ and select the hyperparameters $T_0 = 200$ and $k_0 = 10$. 
We initially collect $50$ datapoints with random excitations and then compare the average estimation error over $200$ runs resulting from the following excitations:
\footnote{The Python code for the numerical example can be accessed at: \url{https://github.com/col-tasas/2025-hidden-convexity-active-learning}} 
\begin{enumerate}
    \item $u_t \simiid \mathcal{N}(0, \frac{\gamma^2}{n_u} I_{n_u})$.
     \item $u_t$ computed by Algorithm~\ref{alg:InputDesignAlgo} with~\cite[Algorithm 3]{wagenmaker2021task}
    \item $u_t$ computed by Algorithm~\ref{alg:InputDesignAlgo} with $\textsc{OptInput} = \mathcal{F}^{-1}\{\check{u}^*\}$ and $\check{u}^*$ as in Theorem~\ref{th:convexRelax}. 
\end{enumerate}
Approaches 2) and 3) require solving problem \eqref{eq:optProbConvex}, for which we use CVXPY~\cite{diamond2016cvxpy,agrawal2018rewriting} with the solver MOSEK~\cite{mosek}.
As shown in Figure~\ref{fig:errorComp}, our proposed input design algorithm outperforms both isotropic random excitations and the approach proposed in \cite{wagenmaker2021task}. 
Specifically, achieving the same accuracy with random excitations requires over 
50\% more samples than with our algorithm.
\begin{figure}[t]
    \centering
    \input{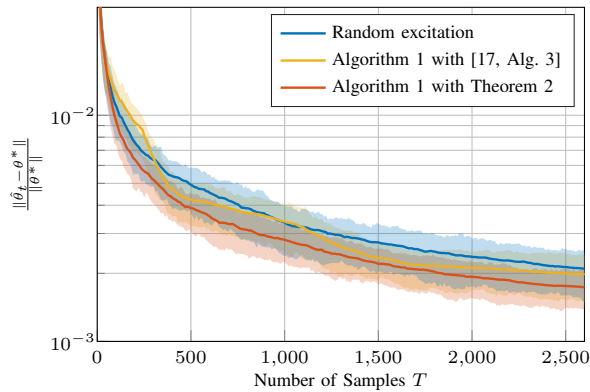}
    \caption{Comparison of the average estimation error over $200$ runs for different excitations. Shaded regions indicate the $25\%$ and $75\%$ percentiles.}
    \label{fig:errorComp}
\end{figure}
Note that in the approach presented in \cite{wagenmaker2021task}, the input sequence associated with each eigenvector is applied for $T_i$ time steps. This results in an episode length that scales with $n_u$. 
Crucially, our approach maintains the original feedback frequency, computing new optimal excitations every $T_i$ rather than every $n_uT_i$ samples. 
This is beneficial since the estimate $\hat \theta_i$ is used to derive the optimal excitation instead of the true system matrix $\theta^*$, and further, the true system is affected by noise. 
On the other hand, when the model used to predict the future evolution of the covariates is exact and there is no noise, our approach will result in the same input sequences as~\cite[Algorithm 3]{wagenmaker2021task}. 
We do not provide a comparison with an input obtained by solving with numerical methods~\eqref{eq:optProbNonConvex}. This is because solving the non-convex problem reliably required significantly more run-time and proved to be very sensitive to the provided initial guess, which cannot be easily obtained in practice.  

\section{Conclusion}
We present an input design algorithm that accelerates identification of ARX systems affected by Gaussian Process noise and provide a finite sample identification error bound.
Furthermore, we propose a convex reformulation of the originally non-convex input design problem, which enables finding the optimal excitation by solving a convex optimization problem. 
Notably, this convex reformulation can also be applied to fully observed state-space systems and hence is of independent interest. 
With this, we obtain a computationally attractive input design algorithm with finite sample guarantees. 
A concluding example showcases the effectiveness of our approach.
Interesting future research directions include the development of finite-sample versions of other classical optimality criteria to complement the existing literature. 

\bibliographystyle{IEEEtran}
\bibliography{references.bib}

\begin{thebibliography}{10}
\providecommand{\url}[1]{#1}
\csname url@samestyle\endcsname
\providecommand{\newblock}{\relax}
\providecommand{\bibinfo}[2]{#2}
\providecommand{\BIBentrySTDinterwordspacing}{\spaceskip=0pt\relax}
\providecommand{\BIBentryALTinterwordstretchfactor}{4}
\providecommand{\BIBentryALTinterwordspacing}{\spaceskip=\fontdimen2\font plus
\BIBentryALTinterwordstretchfactor\fontdimen3\font minus \fontdimen4\font\relax}
\providecommand{\BIBforeignlanguage}[2]{{%
\expandafter\ifx\csname l@#1\endcsname\relax
\typeout{** WARNING: IEEEtran.bst: No hyphenation pattern has been}%
\typeout{** loaded for the language `#1'. Using the pattern for}%
\typeout{** the default language instead.}%
\else
\language=\csname l@#1\endcsname
\fi
#2}}
\providecommand{\BIBdecl}{\relax}
\BIBdecl

\bibitem{wagenmaker2020active}
A.~Wagenmaker and K.~Jamieson, ``Active learning for identification of linear dynamical systems,'' in \emph{Conf. on Learning Theory}.\hskip 1em plus 0.5em minus 0.4em\relax PMLR, 2020, pp. 3487--3582.

\bibitem{bombois2011optimal}
X.~Bombois, M.~Gevers, R.~Hildebrand, and G.~Solari, ``Optimal experiment design for open and closed-loop system identification,'' \emph{Commun. in Information and Systems}, vol.~11, no.~3, pp. 197--224, 2011.

\bibitem{gerencser2009identification}
L.~Gerencs{\'e}r, H.~Hjalmarsson, and J.~M{\aa}rtensson, ``Identification of arx systems with non-stationary inputs—asymptotic analysis with application to adaptive input design,'' \emph{Automatica}, vol.~45, no.~3, pp. 623--633, 2009.

\bibitem{manchester2010input}
I.~R. Manchester, ``Input design for system identification via convex relaxation,'' in \emph{49th IEEE Conf. on Decision and Control (CDC)}, 2010.

\bibitem{Ziemann2023}
I.~Ziemann, A.~Tsiamis, B.~Lee, Y.~Jedra, N.~Matni, and G.~J. Pappas, ``A tutorial on the non-asymptotic theory of system identification,'' in \emph{62nd IEEE Conf. on Decision and Control (CDC)}, 2023.

\bibitem{1977dynamic}
G.~Goodwin and R.~Payne, \emph{Dynamic System Identification: Experiment Design and Data Analysis}, ser. Mathematics in Science and Engineering.\hskip 1em plus 0.5em minus 0.4em\relax Academic Press, 1977.

\bibitem{rojas2007robust}
C.~R. Rojas, J.~S. Welsh, G.~C. Goodwin, and A.~Feuer, ``Robust optimal experiment design for system identification,'' \emph{Automatica}, vol.~43, no.~6, pp. 993--1008, 2007.

\bibitem{gerencser2005adaptive}
L.~Gerencs{\'e}r and H.~Hjalmarsson, ``Adaptive input design in system identification,'' in \emph{44th IEEE Conf. on Decision and Control (CDC)}, 2005.

\bibitem{wainwright2019high}
M.~J. Wainwright, \emph{High-dimensional statistics: A non-asymptotic viewpoint}.\hskip 1em plus 0.5em minus 0.4em\relax Cambridge university press, 2019, vol.~48.

\bibitem{simchowitz2018learning}
M.~Simchowitz, H.~Mania, S.~Tu, M.~I. Jordan, and B.~Recht, ``Learning without mixing: Towards a sharp analysis of linear system identification,'' in \emph{Conf. On Learning Theory}.\hskip 1em plus 0.5em minus 0.4em\relax PMLR, 2018, pp. 439--473.

\bibitem{tsiamis2021linear}
A.~Tsiamis and G.~J. Pappas, ``Linear systems can be hard to learn,'' in \emph{60th IEEE Conf. on Decision and Control (CDC)}, 2021.

\bibitem{sattar2022finite}
Y.~Sattar, S.~Oymak, and N.~Ozay, ``Finite sample identification of bilinear dynamical systems,'' in \emph{61st IEEE Conf. on Decision and Control (CDC)}, 2022.

\bibitem{chatzikiriakos2024b}
N.~Chatzikiriakos, R.~Strässer, F.~Allgöwer, and A.~Iannelli, ``End-to-end guarantees for indirect data-driven control of bilinear systems with finite stochastic data,'' arXiv:2409.18010, 2024.

\bibitem{foster20a}
D.~Foster, T.~Sarkar, and A.~Rakhlin, ``Learning nonlinear dynamical systems from a single trajectory,'' in \emph{Proc. of the 2nd Conf. on Learning for Dynamics and Control}.\hskip 1em plus 0.5em minus 0.4em\relax PMLR, 2020.

\bibitem{Sattar2022}
Y.~Sattar and S.~Oymak, ``Non-asymptotic and accurate learning of nonlinear dynamical systems,'' \emph{J. of Machine Learning Research}, vol.~23, no. 140, pp. 1--49, 2022.

\bibitem{oymak2019non}
S.~Oymak and N.~Ozay, ``Non-asymptotic identification of {LTI} systems from a single trajectory,'' in \emph{2019 Amer. Control Conf. (ACC)}.\hskip 1em plus 0.5em minus 0.4em\relax IEEE, 2019, pp. 5655--5661.

\bibitem{wagenmaker2021task}
A.~Wagenmaker, M.~Simchowitz, and K.~Jamieson, ``Task-optimal exploration in linear dynamical systems,'' in \emph{Int. Conf. on Machine Learning}.\hskip 1em plus 0.5em minus 0.4em\relax PMLR, 2021, pp. 10\,641--10\,652.

\bibitem{mania2022active}
H.~Mania, M.~I. Jordan, and B.~Recht, ``Active learning for nonlinear system identification with guarantees,'' \emph{J. of Machine Learning Research}, vol.~23, no.~32, pp. 1--30, 2022.

\bibitem{lee2024active}
B.~D. Lee, I.~Ziemann, G.~J. Pappas, and N.~Matni, ``Active learning for control-oriented identification of nonlinear systems,'' in \emph{63rd IEEE Conf. on Decision and Control (CDC)}, 2024.

\bibitem{sarkar2019near}
T.~Sarkar and A.~Rakhlin, ``Near optimal finite time identification of arbitrary linear dynamical systems,'' in \emph{Int. Conf. on Machine Learning}.\hskip 1em plus 0.5em minus 0.4em\relax PMLR, 2019, pp. 5610--5618.

\bibitem{karnin2013}
Z.~Karnin, T.~Koren, and O.~Somekh, ``Almost optimal exploration in multi-armed bandits,'' in \emph{Int. Conf. on Machine Learning}.\hskip 1em plus 0.5em minus 0.4em\relax PMLR, 2013, pp. 1238--1246.

\bibitem{diamond2016cvxpy}
S.~Diamond and S.~Boyd, ``{CVXPY}: {A} {P}ython-embedded modeling language for convex optimization,'' \emph{Journal of Machine Learning Research}, vol.~17, no.~83, pp. 1--5, 2016.

\bibitem{agrawal2018rewriting}
A.~Agrawal, R.~Verschueren, S.~Diamond, and S.~Boyd, ``A rewriting system for convex optimization problems,'' \emph{Journal of Control and Decision}, vol.~5, no.~1, pp. 42--60, 2018.

\bibitem{mosek}
\BIBentryALTinterwordspacing
{MOSEK ApS}, \emph{MOSEK Optimizer API for Python 11.0}, 2025. [Online]. Available: \url{https://docs.mosek.com/11.0/pythonapi.pdf}
\BIBentrySTDinterwordspacing

\end{thebibliography}

\end{document}